\documentclass[11pt]{article}

\usepackage[margin=1in]{geometry}
\usepackage{graphicx}
\usepackage[pdftex,colorlinks=true,linkcolor=blue,citecolor=blue,urlcolor=black]{hyperref}
\usepackage{amsmath, amsthm, amssymb}
\usepackage{subfigure}
\usepackage{comment}
\usepackage{url}
\usepackage{xcolor}

\newcommand{\R}{\mathbb{R}}

\newcommand{\E}{\mathbb{E}}

\newcommand{\be}{\begin{equation}}
\newcommand{\ee}{\end{equation}}
\newcommand{\bea}{\begin{eqnarray}}
\newcommand{\eea}{\end{eqnarray}}
\newcommand{\bes}{\begin{equation*}}
\newcommand{\ees}{\end{equation*}}
\newcommand{\beas}{\begin{eqnarray*}}
\newcommand{\eeas}{\end{eqnarray*}}

\newtheorem{thm}{Theorem}
\newtheorem*{thm*}{Theorem}

\newtheorem{lem}[thm]{Lemma}
\newtheorem*{lem*}{Lemma}

\newtheorem{fact}[thm]{Fact}

\newtheoremstyle{definition}
   {}{}{}{0pt}{\bfseries}{.}{.5cm}
   {{\thmname{#1 }}{\thmnumber{#2}}{\thmnote{ (#3)}}}

\theoremstyle{definition}
\newtheorem{dfn}{Definition}

\begin{document}

\title{Almost all decision trees do not allow significant quantum speed-up}

\author{Ashley Montanaro\footnote{Centre for Quantum Information and Foundations, Department of Applied Mathematics and Theoretical Physics, University of Cambridge, UK; {\tt am994@cam.ac.uk}.}}

\maketitle

\begin{abstract}
We show that, for any $d$, all but a doubly exponentially small fraction of decision trees of depth at most $d$ require $\Omega(d)$ quantum queries to be computed with bounded error. In other words, most efficient classical algorithms in the query complexity model do not admit a significant quantum speed-up. The proof is based on showing that, with high probability, the average sensitivity of a random decision tree is high.
\end{abstract}

% ------------------------------------------------------------------------------

\section{Introduction}

Many of the most important examples of quantum algorithms which outperform classical algorithms operate in the query complexity model, where the quantity of interest is the number of queries to the input required to compute some function~\cite{buhrman02}. Can {\em most} classical query algorithms be significantly accelerated by a quantum computer? In one sense, the answer is a resounding no. It was shown by Ambainis in 1999 that almost all boolean functions on $n$ bits require $n/4 - O(\sqrt{n} \log n)$ quantum queries to be computed with bounded error~\cite{ambainis99a}; using different techniques, Ambainis et al have very recently improved this bound to $n/2 - o(n)$~\cite{ambainis12a}.

However, random boolean functions are arguably somewhat uninteresting in general in that they are hard both for classical and quantum algorithms to compute. One can interpret the result of \cite{ambainis99a} as saying that there are too many boolean functions, and too few efficient algorithms, whether classical or quantum. Perhaps a more interesting question is whether most functions with efficient classical algorithms in the query model -- i.e.\ short decision trees -- could be computed significantly more quickly on a quantum computer than is possible classically. In this note we show that this is also not the case.

We first need to define the model of decision trees which we will use (see~\cite{buhrman02} for further background). A (boolean) decision tree $T$ is a rooted binary tree where each internal vertex has exactly two children. Such trees are called full; only full binary trees will be considered in this paper. Each vertex is labelled with a variable $x_i$, $1 \le i \le n$, and each leaf is labelled with 0 or 1, corresponding to the output of the tree. Each edge from a node labelled with $x_i$ to its children is labelled with $0$ or $1$, corresponding to the value of $x_i$. $T$ computes a boolean function $T:\{0,1\}^n \rightarrow \{0,1\}$ in the obvious way: starting with the root, the variable labelling each vertex is queried, and dependent on whether the answer is 0 or 1 the left or right subtree is evaluated. When a leaf is reached, the tree outputs the label of that leaf. The depth of a vertex is defined as follows: the root has depth 0, and the depth of any other vertex is equal to the depth of its parent, plus 1. The depth of $T$ is the largest depth of any vertex, i.e.\ the worst-case number of queries made on any input.%A pre-leaf node is a node whose children are both leaves.

Following~\cite{jackson05}, we will consider the {\em uniform} model of random decision trees. In this model, we obtain a random depth $d$ decision tree $T$ on $n \ge d$ variables simply by picking $T$ uniformly at random from the set of all non-redundant decision trees of depth at most $d$ on $n$ variables. A decision tree is said to be non-redundant if no variable occurs more than once on any path from the root to a leaf. Observe that choosing a boolean function by picking a tree uniformly at random does not pick the corresponding function itself uniformly at random from the set of all functions with decision trees of depth $d$, as many different decision trees can represent the same function, and some functions may have more trees representing them than others. Also observe that the internal structure of a uniformly random decision tree $T$ and the values assigned to its leaves are independent. In other words, a random decision tree $T$ of depth at most $d$ on $n$ variables can be obtained by a two-stage process: first, pick the structure and labels of the internal vertices of $T$ uniformly at random from the set of all decision trees of depth at most $d$ on $n$ variables, then assign 0 or 1 to the leaves of $T$ uniformly at random.

One could also consider an alternative model of random decision trees, known as the {\em complete} model~\cite{jackson05}. In this model, $T$'s structure is always that of a complete binary tree of depth $d$. We randomly assign variables to $T$'s internal vertices, consistent with $T$ being non-redundant, and assign 0 or 1 to $T$'s leaves uniformly at random. This model appears somewhat more artificial and we therefore concentrate on the uniform model here; however, our results easily extend to the complete model.

We can now state our main result.

\begin{thm}
\label{thm:boundederr}
Let $T$ be a random depth $d$ decision tree on $n$ variables. Then there exists a universal constant $\alpha > 0$ such that $\Pr[Q_2(f) < \alpha d] \le 2^{-\Omega(2^{d/3})}$.
\end{thm}

In the above theorem, $Q_2(f)$ is the bounded-error quantum query complexity of $f$~\cite{buhrman02}, i.e.\ the smallest number of quantum queries required to compute $f$ with worst-case success probability at least $2/3$. Observe that the concentration bound obtained is doubly exponentially small in $d$. This result thus implies that the vast majority of efficient classical algorithms in the query complexity model do not admit a significant quantum speed-up.

The counting technique used by Ambainis~\cite{ambainis99} to prove that most boolean functions do not have query-efficient quantum algorithms does not seem to suffice to prove Theorem \ref{thm:boundederr}. Indeed, Ambainis showed that there are at most $2^{O(n^{2d+3} d)}$ boolean functions on $n$ bits that can be computed by quantum algorithms making $d$ queries. However, there can be $n^{\Theta(2^d)}$ decision trees of depth $d$. Thus this argument does not disallow, for example, the possibility that random decision trees of depth $d = O(\log^2 n)$ could be evaluated using $O(\log n)$ quantum queries.

We therefore take another approach to proving the above theorem, by giving a lower bound on the {\em average sensitivity} (also known as total influence) of random decision trees. This quantity is mathematically tractable and is known to lower bound quantum query complexity~\cite{shi00}.

% ------------------------------------------------------------------------------

\section{Proof of Theorem \ref{thm:boundederr}}
\label{sec:proof}

The average sensitivity of a boolean function $f:\{0,1\}^n \rightarrow \{0,1\}$ is defined as
\[ \bar{s}(f) = \frac{1}{2^n} \sum_{x \in \{0,1\}^n} \sum_{i=1}^n |f(x) - f(x^i)|, \]
where $x^i$ is the $n$-bit string obtained by flipping the $i$'th bit of $x$. Up to normalisation, $\bar{s}(f)$ thus counts the number of neighbours $x$, $y$ such that $f(x) \neq f(y)$. Average sensitivity gives the following lower bound on quantum query complexity.

\begin{thm}[Shi~\cite{shi00}, see~\cite{dewolf08} for the version here]
\label{thm:shi}
Assume that $f:\{0,1\}^n \rightarrow \{0,1\}$ can be computed by a quantum algorithm making $q$ queries and with worst-case failure probability $\epsilon \le 1/2$. Then $q \ge \frac{1}{2}(1-2\epsilon)^2 \bar{s}(f)$.
\end{thm}

Our goal will now be to upper bound the probability that the average sensitivity of a random decision tree is low. To do so, we will use the following powerful measure concentration result.

\begin{thm}[See e.g.\ \cite{ledoux01} or {\cite[Corollary 5.2]{dubhashi09}}]
\label{thm:measconc}
Fix $\eta>0$ and assume that $g:\{0,1\}^n \rightarrow \R$ satisfies $|g(x)-g(y)| \le \eta$ for all $x,y \in \{0,1\}^n$ such that $d(x,y)=1$, where $d(x,y)$ is the Hamming distance between $x$ and $y$. Let $x \in \{0,1\}^n$ be picked uniformly at random. Then
\[ \Pr_x\left[ g(x) < \E_x[g(x)] - \delta\right] \le e^{-2\delta^2/(n \eta^2)}. \]
\end{thm}

We first calculate the expected average sensitivity of random decision trees. As discussed previously, we will choose a uniformly random decision tree by first choosing the structure of the tree, then choosing assignments to the leaves uniformly at random. So fix a decision tree $T$ on $n$ variables and let $\mathcal{L}$ be the set of leaves of $T$, setting $L := |\mathcal{L}|$. Let $T_z:\{0,1\}^n \rightarrow \{0,1\}$ be the boolean function obtained by assigning $z \in \{0,1\}^L$ to the leaves of $T$; we will eventually apply Theorem \ref{thm:measconc} to the function $g:\{0,1\}^L \rightarrow \R$ defined by $g(z) = \bar{s}(T_z)$. %and random (but fixed) values to the leaves of $T$ not in $\mathcal{L}$, 
%and consider $g:\{0,1\}^L \rightarrow \R$ defined by $g(z) = \bar{s}(T_z)$.
For any leaf $\ell$, let $d(\ell)$ be the depth of $\ell$.

\begin{lem}
\label{lem:expas}
\[ \E_{z\in\{0,1\}^L}[\bar{s}(T_z)] = \frac{1}{2} \sum_{\ell \in \mathcal{L}} d(\ell) 2^{-d(\ell)}. \]
\end{lem}

\begin{proof}
From the definition of $g$,
\[ \E_{z\in\{0,1\}^L}[\bar{s}(T_z)] = \frac{1}{2^n} \sum_{x \in \{0,1\}^n} \sum_{i=1}^n \E_{z\in\{0,1\}^L}[|T_z(x) - T_z(x^i)|]. \]
Now observe that for each $i$, if $T$ queries the $i$'th bit on input $x$, $\E_{z\in\{0,1\}^L}[|T_z(x) - T_z(x^i)|] = \frac{1}{2}$. Otherwise, $\E_{z\in\{0,1\}^L}[|T_z(x) - T_z(x^i)|] = 0$. So we obtain
\[ \E_{z\in\{0,1\}^L}[\bar{s}(T_z)] = \frac{1}{2^{n+1}} \sum_{x \in \{0,1\}^n} |\{i: T \text{ queries $i$ on }x\}| = \frac{1}{2} \E_{x \in \{0,1\}^n} [\text{path length on input }x]. \]
As each query to the input gives 0 or 1 with equal probability on a random input $x$, the probability of ending up at a given leaf $\ell$ on a random input is just $2^{-d(\ell)}$. The claim follows.
\end{proof}

We can also bound $\eta = \max_{d(w,z)=1} |\bar{s}(T_w)-\bar{s}(T_z)|$ as follows.

\begin{lem}
\label{lem:lipschitz}
Let $w,z \in \{0,1\}^L$ satisfy $d(w,z)=1$. Then
\[ |\bar{s}(T_w)-\bar{s}(T_z)| \le \max_{\ell \in \mathcal{L}} d(\ell) 2^{1-d(\ell)}. \]
\end{lem}

\begin{proof}
The idea is essentially the same as the proof of Lemma \ref{lem:expas}. We have
\[ \bar{s}(T_w) - \bar{s}(T_z) = \frac{1}{2^n} \sum_{x \in \{0,1\}^n} \sum_{i=1}^n \left( |T_w(x) - T_w(x^i)| - |T_z(x) - T_z(x^i)| \right). \]
As $w$ and $z$ only differ on one leaf $\ell$, $|T_w(x) - T_w(x^i)| = |T_z(x) - T_z(x^i)|$ unless precisely one of $x$ and $x^i$ leads to $\ell$. Thus
\[ |\bar{s}(T_w) - \bar{s}(T_z)| \le \frac{1}{2^n} \sum_{x,x \rightarrow \ell} \sum_{i, x^i \nrightarrow \ell} 1 + \frac{1}{2^n} \sum_{x,x \nrightarrow \ell} \sum_{i, x^i \rightarrow \ell} 1 = \frac{1}{2^{n-1}} \sum_{x,x \rightarrow \ell} |\{i:x^i \nrightarrow \ell\}|, \]
where we use the notation $x \rightarrow \ell$ and $x \nrightarrow \ell$ to mean that evaluation of the tree does or does not end up in leaf $\ell$ on input $x$, respectively. For each $x$, there can only be at most $d(\ell)$ variables $i$ such that $x \rightarrow \ell$ and $x^i \nrightarrow \ell$. Therefore
\[ |\bar{s}(T_w) - \bar{s}(T_z)| \le 2d(\ell) \Pr_{x \in \{0,1\}^n} [x \rightarrow \ell] = d(\ell) 2^{1-d(\ell)}. \]
\end{proof}

Note that this bound is quite weak when there exists a leaf $\ell$ in $T$ of low depth. Nevertheless, its dependence on $d(\ell)$ is essentially tight: consider the tree $T$ that computes
\[ T(x) = \alpha x_1 + (1-x_1)x_2 x_3 \dots x_n, \]
for $\alpha \in \{0,1\}$. Then changing $\alpha$ from 0 to 1 changes $\bar{s}(T)$ from close to 0 to close to 1. Luckily, it turns out that in fact $\min_{\ell \in \mathcal{L}} d(\ell)$ will be high with high probability; roughly speaking, a random tree is unlikely to have any low-hanging fruit. We formalise this as the following lemma, which we prove at the end.

\begin{lem}
\label{lem:lowdepth}
Let $T$ be picked uniformly at random from the set of decision trees of depth at most $d$. Then, for any $h \le d - \log_2 d - 2$, the probability that $T$ has a leaf with depth at most $h$ is at most $2^{1-2^{d-h-2}}$.
\end{lem}

Assume that $d(\ell) \ge 2d/3$ for all $\ell \in \mathcal{L}$. Then, using Theorem \ref{thm:measconc} and Lemma \ref{lem:lipschitz}, we have
\beas
\Pr_z \left[\bar{s}(T_z) < \E[ \bar{s}(T_z) ] - \delta\right] \le e^{-2\delta^2/(L \eta^2)} \le e^{-(9/8)2^{d/3}\delta^2 / d^2}.
\eeas
Also, if $d(\ell) \ge 2d/3$ for all $\ell \in \mathcal{L}$, by Lemma \ref{lem:expas} $\E[ \bar{s}(T_z) ] \ge d/3$. Thus
\[ \Pr_z [ \bar{s}(T_z) < (1-\epsilon) d/3 ] \le e^{-2^{d/3-3} \epsilon^2}. \]
Now let $T$ be picked uniformly at random from the set of decision trees of depth at most $d$. Taking $h=2d/3$ in Lemma \ref{lem:lowdepth} and using a union bound over the two bad events that $\min_{\ell \in \mathcal{L}} d(\ell) \le 2d/3$ and $\bar{s}(T_z) < (1-\epsilon) d/3$, we obtain the overall bound that
\[ \Pr_T [ \bar{s}(T) < (1-\epsilon) d/3 ] \le 2^{1-2^{d/3-2}} + e^{-2^{d/3-3} \epsilon^2}, \]
which is clearly of order $2^{-\Omega(2^{d/3} \epsilon^2)}$. By Theorem \ref{thm:shi}, $Q_2(T) \ge \bar{s}(T)/18$. Taking $\epsilon$ to be an arbitrary constant such that $0<\epsilon<1$ completes the proof of Theorem \ref{thm:boundederr}.

Observe that, if one is content with looser concentration bounds, it is possible to show that $\Pr_T [ \bar{s}(T) < d/2 - O(\log d) ] = 2^{-\Omega(d)}$, by taking $h=d-O(\log d)$ in Lemma \ref{lem:lowdepth}.

% ------------------------------------------------------------------------------

\subsection{Remaining lemmas}

In order to prove Lemma \ref{lem:lowdepth} we will need a simple bound on the number of low depth binary trees. 

\begin{fact}
\label{fact:counting}
Let $N_d$ be the number of binary trees of depth at most $d$. For $d \ge 1$, $N_d = N_{d-1}^2 + 1$, and $N_0=1$. Thus $N_d \ge 2^{2^{d-1}}$ for all $d \ge 1$.
\end{fact}

\begin{proof}
Every binary tree of depth at most $d$ either consists of a single leaf or of two independent binary trees of depth at most $d-1$; the claimed recurrence for $N_d$ is immediate. The second part easily follows from this recurrence by induction.
\end{proof}

\begin{lem}
\label{lem:lowdepth2}
Let $T$ be picked uniformly at random from the set of binary trees of depth at most $d$. Then, for any $h \le d - \log_2 d - 2$, the probability that $T$ has a leaf with depth at most $h$ is at most $2^{1-2^{d-h-2}}$.
\end{lem}

\begin{proof}
Using a union bound and Fact \ref{fact:counting}, the probability that $T$ has a leaf with depth at most $h$ is upper bounded by
\[ \sum_{k=0}^h 2^k \Pr[\text{a vertex at depth $k$ is a leaf}] = \sum_{k=0}^h \frac{2^k}{N_{d-k}} \le \sum_{k=0}^h 2^{k - 2^{d-k-1}} \le \sum_{k=0}^h 2^{-2^{d-k-2}} \le 2^{1-2^{d-h-2}}. \]
\end{proof}

Picking a decision tree $T$ uniformly at random from the set of decision trees of depth at most $d$ is similar to, but not quite the same as, making $T$'s structure a uniformly random binary tree of depth at most $d$, then assigning variables to the internal vertices at random, and values to the leaves at random. However, the above conclusions also apply to uniformly random decision trees (in other words, Lemma \ref{lem:lowdepth2} implies Lemma \ref{lem:lowdepth}). The reason is that choosing a tree from a distribution where trees with different variables assigned to the internal vertices are considered to be distinct can only decrease the probability that an arbitrary vertex at a given depth is a leaf.

% ------------------------------------------------------------------------------

\section{Conclusions}

We have shown that, in the query complexity model, most short decision trees do not admit a quantum speed-up greater than a constant factor. An interesting open question is to determine the extent of the quantum speed-up which {\em is} possible. It is known that, for all boolean functions $f:\{0,1\}^n \rightarrow \{0,1\}$, $Q_2(f) \le n/2 + O(\sqrt{n})$~\cite{vandam98}. Could it be the case that all decision trees of depth $d$ can be computed with bounded error using $d/2 + O(\sqrt{d})$ quantum queries? Indeed, one can show that a non-zero speed-up is possible for essentially all decision trees: if $f$ has a decision tree of depth $d \ge 2$, then $Q_2(f) \le d-1$. The reason is that every boolean function on 2 bits can be computed with success probability $9/10$ using only one quantum query~\cite{kerenidis04}, so the quantum algorithm can simply follow the classical decision tree, replacing the last two classical queries with one quantum query.

% ------------------------------------------------------------------------------

\section*{Acknowledgements}

This work was supported by an EPSRC Postdoctoral Research Fellowship. I would like to thank Tony Short for his interpretation of Lemma \ref{lem:lowdepth}, and Ronald de Wolf for helpful comments on a previous version.

% ------------------------------------------------------------------------------

\bibliographystyle{plain}
\bibliography{../thesis}

% ------------------------------------------------------------------------------
\end{document}